\documentclass[12pt,a4paper]{amsart}
\usepackage{amsmath,amsfonts,amssymb}
\usepackage{color}
\usepackage
{hyperref}

\numberwithin{equation}{section}

\newcommand{\dd}{\partial}
\newcommand{\f}{\frac}

\DeclareMathOperator{\rk}{rk}

\newtheorem{theorem}{Theorem}[section]

\newtheorem{corollary}[theorem]{Corollary}
\newtheorem{proposition}[theorem]{Proposition}

\newtheorem*{AffClass}{Theorem: affine classification}
\newtheorem*{CompTh}{Compatibility Theorem}
\newtheorem*{CorrTh}{Projective Correspondence Theorem}

\theoremstyle{definition}
\newtheorem{definition}[theorem]{Definition}

\theoremstyle{remark}
\newtheorem{remark}[theorem]{Remark}

\numberwithin{equation}{section}

\newcommand*{\pd}
[2]{\mathchoice{\frac{\partial#1}{\partial#2}}
  {\partial#1/\partial#2}{\partial#1/\partial#2}
  {\partial#1/\partial#2}}

\def\d{\partial}

\begin{document}
\title[Bi-Hamiltonian structures of KdV type]{
  Bi-Hamiltonian structures of KdV type,\\ cyclic Frobenius algebrae\\ and
  Monge metrics} \author{Paolo Lorenzoni}
\address{P.~Lorenzoni
  \\
  Dipartimento di Matematica e Applicazioni, Universit\`a di Milano-Bicocca,
  \\
  Via Roberto Cozzi 55, I-20125 Milano, Italy and INFN sezione di
  Milano-Bicocca}
\email{paolo.lorenzoni@unimib.it}
\author{Raffaele Vitolo}
\thanks{This research has been partially supported by the Department of
  Mathematics and Applications of the Universit\`a di Milano-Bicocca,
  Department of Mathematics and Physics ``E. De Giorgi'' of the Universit\`a
  del Salento, GNFM of the Istituto Nazionale di Alta Matematica (INdAM), the
  research project Mathematical Methods in Non Linear Physics (MMNLP) by the
  Commissione Scientifica Nazionale -- Gruppo 4 -- Fisica Teorica of the
  Istituto Nazionale di Fisica Nucleare (INFN). P L is supported by funds of H2020-
MSCA-RISE-2017 Project No. 778010 IPaDEGAN}
\address{R.~Vitolo\\
  Dipartimento di Matematica e Fisica ``E. De Giorgi''
  Universit\`a del Salento
  \\
  via per Arnesano, 73100 Lecce, Italy and
  INFN sezione di Lecce}
\email{raffaele.vitolo@unisalento.it}

\begin{abstract}
  We study algebraic and projective geometric properties of Hamiltonian trios
  determined by a constant coefficient second-order operator and two
  first-order localizable operators of Ferapontov type. We show that
  first-order operators are determined by Monge metrics, and define a structure
  of cyclic Frobenius algebra. Examples include the AKNS system, a
  $2$-component generalization of Camassa-Holm equation and the Kaup--Broer
  system. In dimension $2$ the trio is completely determined by two conics of
  rank at least $2$. We provide a partial classification in dimension $4$.
\end{abstract}

\maketitle

\section{Introduction}

It was observed in \cite{OR} that many important bi-Hamiltonian structures of
integrable (systems of) PDEs have the form
\begin{equation}\label{kdvtype}
  (P_1,Q_1+\epsilon^kR_{k+1}),
\end{equation}
(no sum over $k$) where $P_1$ and $Q_1$ are first order compatible homogeneous
Hamiltonian operators (Hamiltonian operators of hydrodynamic type) and
$R_{k+1}$ is a single $(k+1)$-th order omogeneous Hamiltonian operator
compatible with $P_1$ and $Q_1$. Here, the homogeneity is defined with respect
to the grading $\operatorname{deg}\partial_x=1$.

Denoting with the square bracket the Schouten bracket we have
\begin{displaymath}
  [P_1,Q_1]=[P_1,R_{k+1}]=[Q_1,R_{k+1}]=0.
\end{displaymath}
In other words the building blocks of the pair \eqref{kdvtype}
$(P_1,Q_1,R_{k+1})$ define a \emph{trio} of Hamiltonian structures.  The above
structure can be thought as a deformation of the bi-Hamiltonian structure of
hydrodynamic type $(P_1,Q_1)$. Due to the general theory of deformations the
most interesting cases are $k=1$ and $k=2$ since for $k>2$ the deformation
$R_{k+1}$ can be always eliminated by Miura type transformations \cite{LZ}.
The most famous example of such structures is the Hamiltonian trio
\begin{equation}
  P=P_1=\d_x,\qquad Q=Q_1+R_3,\quad Q_1=2u\d_x+u_x,\quad R_3=\d_x^3.
\end{equation}
Coupling $Q_1$ and $R_3$ one obtains the bi-Hamiltonian structure of the KdV
hierarchy
\begin{equation}
  (\d_x,2u\d_x+u_x+\epsilon^2\d_x^3)
\end{equation}
discovered by Magri in \cite{magri}, while coupling $P_1$ and $R_3$ one obtains
the bi-Hamiltonian structure of the Camassa--Holm hierarchy
\begin{equation}
  (2u\d_x+u_x,\d_x+\epsilon^2\d_x^3).
\end{equation}
Bi-Hamiltonian structures \eqref{kdvtype} obtained in this way have been introduced in \cite{OR} and have been called
in \cite{LSV:bi_hamil_kdv} \emph{bi-Hamiltonian structures of KdV type}.
Another example (from \cite{LZ,F}) is the trio:
\begin{equation}
  P_1=\begin{pmatrix}
    0 & \d_x\\ \d_x & 0 
  \end{pmatrix},\,\, Q_1=
  \begin{pmatrix}
    2u\partial_x+u_x &  v\d_x\\
    \d_x v & -2\d_x
  \end{pmatrix},\,\, R_2=\begin{pmatrix}
    0 &  -\d_x^2\\
    \d_x^2 & 0
  \end{pmatrix}
\end{equation}
In this case one coupling yields the bi-Hamiltonian structure of the the
so-called AKNS hierarchy, and the other one yields the bi-Hamiltonian structure
of the two component Camassa-Holm hierarchy \cite{LZ,F}.

In order to classify this kind of bi-Hamiltonian structures (with $k=1,2$) one
can use the following strategy:
\begin{enumerate}
\item Use canonical forms of $R_{k+1}$ under some natural groups of
  transformations preserving the form of $(P_1,Q_1,R_{k+1})$.
\item Compute compatibility conditions $[P_1,R_{k+1}]=0$.
\item Use compatibility conditions to obtain trios of Hamiltonian operators.
\end{enumerate}
The above strategy is motivated by the fact that there exist classifications of
canonical forms of operators $R_{k+1}$ under the action of various
transformation groups, while trying to work with canonical forms of $P_1$ or
$Q_1$ does not lead to manageable forms of the corresponding $R_{k+1}$, in view
of the greater complexity of the latter.

There are two natural choices of the group of transformations to deal with: the
group of diffeomorphisms of the dependent variables and the groups of
reciprocal projective transformations of the independent variables. Depending
on the choice of the group the problem admits a slightly different
formulation. In the first case, since the group of diffeomorphisms preserves
the locality of Hamiltonian operators it is possible to restrict the attention
only to \emph{local} first order Hamiltonian operators (also known as
Dubrovin-Novikov Hamiltonian operators)
\begin{equation}
  \label{eq:5}
  P_1 = g^{ij}\d_x + \Gamma^{ij}_ku^k_x.
\end{equation}
In the second case, we can use the group of transformations of dependent
variables to reduce the operator $R_{k+1}$ to Doyle--Potemin canonical form:
(see \cite{lorenzoni23:_miura_poiss} and references therein)
\begin{equation}
  \label{eq:36}
  R_{k+1} = \partial_x\circ R_{k-1}\circ\partial_x,
\end{equation}
where $R_{k-1}$ is a homogeneous operator of order $k-1$. Reciprocal projective
transformations preserve this form \cite{lorenzoni23:_miura_poiss}; however,
they do not preserve locality of $P_1$, so that one is obliged to consider
first-order Hamiltonian operators of localizable shape (or simply localizable)
\begin{equation}
  \label{eq:6}
  P_1=  g^{ij}\dd_{x} +
  \Gamma^{ij}_{k }u^k_{x}
  + w^i_ku^k_{x}\partial_{{x}}^{-1}{u}^j_{{x}}
  + {u}^i_{{x}}\partial_{{x}}^{-1}w^j_k{u}^k_{{x}},
\end{equation}
The first approach has been pursued in \cite{LSV:bi_hamil_kdv} using the
results of \cite{doyle93:_differ_poiss,GP87} for second-order operators $R_2$
and the results of \cite{GP91,GP97,doyle93:_differ_poiss,FPV14,FPV16} for
third-order operators $R_3$.  For instance, in the $2$-component case the
canonical forms are
\begin{eqnarray}
  \label{eq:112}
  R_2 &=& \begin{pmatrix}0 & 1\\ -1 & 0\end{pmatrix}\d_x^2,\\
  \label{P1}
  R^{(1)}_3&=&
               \begin{pmatrix}
                 0 & 1\\
                 1 & 0
               \end{pmatrix}\d_x^3,\\
  \label{P2}
  R^{(2)}_3&=& \d_x
               \begin{pmatrix}
                 0 & \d_x\frac{1}{u^1} \\
                 \frac{1}{u^1}\d_x&
                 \frac{u^2}{(u^1)^{2}}\d_x+\d_x\frac{u^2}{(u^1)^{2}}
               \end{pmatrix}
                                    \d_x,\\
  \label{P3}
  R^{(3)}_3&=& \d_x
               \begin{pmatrix}
                 \d_x & \d_x \frac{u^2}{u^1}\\
                 \frac{u^2}{u^1} \d_x & \frac{(u^2)^2+1}{2(u^1)^2}\d_x+
                 \d_x\frac{(u^2)^2+1}{2(u^1)^{2}}
               \end{pmatrix}\d_x.
\end{eqnarray}
and the corresponding compatible first order operators are given in the
following theorem.
\begin{AffClass}[\cite{LSV:bi_hamil_kdv}]\label{th1}
  $P_1$ is a Hamiltonian operator compatible with $R_2$ if and only if
  \begin{subequations}\label{eq:121}
    \begin{align}
      &g^{11}=c_1u^1+c_2,
      \\
      &g^{12}=\f{1}{2}c_3u^1+\f{1}{2}c_1u^2+c_5
      \\
      &g^{22}=c_3u^2+c_4.
    \end{align}
  \end{subequations}
  $P_1$ is a Hamiltonian operator compatible with $R_3^{(1)}$ if and only if
  \begin{subequations}\label{eq:221}
    \begin{align}
      g^{11} =& c_1 u^1 + c_2 u^2 + c_3,
      \\
      g^{12} =& c_4 u^1 + c_1 u^2 + c_5
      \\
      g^{22}= & c_6 u^1 + c_4 u^2 + c_7
    \end{align}
  \end{subequations}
  together with the algebraic conditions
  \begin{equation}\label{eq:165}
    c_1c_4 - c_2c_6=0,
    \quad
    c_3c_4 - c_7c_2=0,
    \quad
    c_3c_6 - c_1c_7=0.
  \end{equation}
  $P_1$ is a Hamiltonian operator compatible with $R_3^{(2)}$ if and only if
  \begin{subequations}\label{eq:145}
    \begin{align}
      &g^{11}=c_1 u^1 + c_2 u^2,
      \\
      &g^{12}=c_4 u^1 + \frac{c_3}{u^1} + \frac{c_2 (u^2)^2}{2 u^1},
      \\
      &g^{22}=2 c_4 u^2 + \frac{c_6}{ u^1} - \frac{c_1 (u^2)^2}{ u^1} + c_5,
    \end{align}
  \end{subequations}
  together with the algebraic conditions
  \begin{equation}\label{eq:172}
    c_2 c_6 + 2 c_1 c_3 = 0, \quad c_2 c_5 = 0, \quad c_1 c_5 = 0.
  \end{equation}
  $P_1$ is a Hamiltonian operator compatible with $R_3^{(3)}$ if and only if
  \begin{subequations}\label{eq:151}
    \begin{align}
      &g^{11}=c_1 u^1 + c_2 u^2+c_3,
      \\
      &g^{12}=c_4 u^1 - \frac{c_2}{2u^1} + \frac{c_3 u^2}{ u^1}+ \frac{c_2
        (u^2)^2}{ 2u^1},
      \\
      &g^{22}=2 c_4 u^2 + \frac{c_1}{u^1} + \frac{c_5 u^2}{ u^1}- \frac{c_1
        (u^2)^2}{ u^1} +c_6,
    \end{align}
  \end{subequations}
  together with the algebraic conditions
  \begin{equation}\label{eq:182}
    c_2 c_5 + 2 c_1 c_3 = 0, \quad c_2 c_6 - 2 c_3 c_4 = 0,
    \quad c_1 c_6 + c_4 c_5 = 0.
  \end{equation}
\end{AffClass}
The above families of contravariant metrics depend linearly on the parameters
and thus any pair of metrics belonging to these families defines a
bi-Hamiltonian structure of hydrodynamic type compatible with the second/third
order operator.

The second approach has been pursued in \cite{LV} in the case
$R_3=\eta^{ij}\partial_x^3$ where $(\eta^{ij})$ is a symmetric constant
non-degenerate matrix ($\det(\eta^{ij})\neq 0$). The operator $R_3$ generates
the simplest orbit of third-order operators under the action of the projective
reciprocal transformation group. The study of compatibility conditions leads to
the following results:
\begin{itemize}
\item the Christoffel symbols $\Gamma^{ij}_k$ define a Frobenius algebra
  structure on the cotangent bundle of the manifold of dependent variables
  $(u^i)$;
\item the operator $P=L+N$ splits into its local part $L$ and its non-local part
  $N$, and they are independently Hamiltonian operators;
\item for $n\geq 3$, $N=0$, and the operator becomes purely local.
\end{itemize}
The local trios that we got for $n>2$ are known in the literature. They can be
obtained as special cases of the results of \cite{strachan14:_novik_camas} and
also were studied in \cite{bolsinov:_applic_nijen_iii} in  terms of Frobenius pencils. However we point out that the
locality (for $n>2$) is not an \emph{a priori} assumption but the result of non
trivial computations.

The aim of the present paper is to study the case $R_2=\eta^{ij}\partial_x^2$
where $(\eta^{ij})$ is a skew-symmetric constant non-degenerate matrix
($\det(\eta^{ij})\neq 0$). The operator $R_2$ generates the simplest orbit of
second-order operators under the action of the projective reciprocal
transformation group.  The role of Frobenius algebra in this setting is played
by a new type of algebra recently introduced by Buchstaber and Mikhailov and
called \emph{cyclic Frobenius algebra}.

\begin{definition}\cite{buchstaber23:_cyclic_froben}
  Let $\mathcal{V}$ be some $\mathbb{C}$-linear space
  (dim($\mathcal{V} \ge 1$).  A cyclic Frobenius algebra (CF-algebra)
  $\mathcal{A}$ is an associative algebra $\mathcal{A}$ with unity $1$ equipped
  with a $\mathbb{C}$-bilinear skew-symmetric form
  $\eta(\cdot, \cdot):\mathcal{A}
  \otimes_{\mathbb{C}}\mathcal{A}\to\mathcal{V}$ such that
  \begin{equation}\label{cyclic}
    \eta(A, B\circ C)+\eta(B, C\circ A)+\eta(C, A\circ B) = 0
  \end{equation}
  where $A, B, C\in\mathcal{A}$ and $\circ$ is the product in the algebra.
\end{definition}

Let $\mathcal{V}=\mathbb{C}$. Denoting by $\Gamma^{ij}_k$ the structure
constants of the product we have
\begin{multline*}
  \eta^{ij}A_i(B\circ C)_j+\eta^{ij}B_i(C\circ A)_j+\eta^{ij}C_i(A\circ B)_j
  =
  \\
  (\eta^{ij}\Gamma^{lk}_j+\eta^{lj}\Gamma^{ki}_j
  +\eta^{kj}\Gamma^{il}_j)A_iB_lC_k=  0
\end{multline*}
or, taking into account that $A,B,C$ are arbitrary,
\begin{equation}\label{eq:57}
  \Gamma^{ki}_j\eta^{jl}+\Gamma^{il}_j\eta^{jk}+\Gamma^{lk}_j\eta^{ji}=0.
\end{equation}

The main results of the paper concerning the compatibility of Hamiltonian
operators can be summarized as follows.
\begin{CompTh}\label{sec:cond-comp}
  The Hamiltonian operators $P$, $R$ are compatible if and only if
  \begin{itemize}
  \item $w^i_j=W^i_j$ where $W$ is a constant matrix that is symmetric with
    respect to $\eta$:
    \begin{gather}
      \label{cond000}
      \eta(AW,B)=\eta(A,BW);
    \end{gather}
  \item the contravariant Christoffel symbols are linear functions of the form
    \begin{equation}
      \Gamma^{ij}_k=\partial_k(-W^j_su^su^i+b^{ij}_su^s)=
      - W^j_ku^i - W^j_su^s\delta^i_k + b^{ij}_k
    \end{equation}
    for certain constants $b^{ij}_k$.
  \item the product with structure constants $\Gamma^{ij}_k$
    \begin{displaymath}
      (A\circ B)_i=\Gamma^{jk}A_jB_k
    \end{displaymath}
    endows the cotangent space $T^*M$ of the manifold $M$ of dependent
    variables with a structure of \emph{cyclic Frobenius algebra} (without
    unity) and satisfy the conditions:
    \begin{equation}\label{cocycle}
      \eta(A\circ B,C)=\eta(A,C\circ B).
    \end{equation}
  \end{itemize}
\end{CompTh}
The Theorem is stated and proved as Theorem~\ref{th:cond-comp}.

Notice that the condition \eqref{cocycle} can be also written as
\begin{equation}\label{cocycle2}
  \Gamma^{ij}_l\eta^{lk}+\Gamma^{kj}_l\eta^{li}=0.
\end{equation}
Indeed, relabelling the indices the condition
\begin{displaymath}
  \eta^{ij}(A\circ B)_i=\eta^{ij}A_i(C\circ B)_j
\end{displaymath}
reads
\begin{displaymath}
  (\Gamma^{lk}_i\eta^{ij}+\Gamma^{jk}_i\eta^{il})A_lB_kC_j=0.
\end{displaymath}
Conditions \eqref{cyclic} and \eqref{cocycle} appear in the paper
\cite{strachan14:_novik_camas} as cocycle conditions arising from the
compatibility between a local first order Hamiltonian operator of hydrodynamic
type defined by a flat linear metric and a second order constant Hamiltonian
operator defined by a skew-symmetric matrix. In this setting the contravariant
Christoffel symbols of the linear metric are constant and define the structure
constants of a Balinsky-Novikov algebra.

A corollary of the above theorem is that the nondegenerate symmetric bilinear
form obtained from the the contravariant metric defining $P_1$ ``lowering'' the
indices with $\eta$:
\begin{equation}\label{eq:62}
  \bar{g}_{ab}=\eta_{jb}\eta_{ia}g^{ij}
\end{equation}
is the Monge metric of a \emph{quadratic line complex}, an algebraic variety
that is defined in the Pl\"ucker embedding of the projective space with
homogeneous coordinates $[v^1,\ldots,v^{n+1}]$, where $u^i=v^i/v^{n+1}$,
$i=1,\ldots,n$, $v^{n+1}\neq 0$. See \cite{FPV14,FPV16} for more details on
this construction. The equation that characterizes Monge metrics is
\begin{equation}
  \label{eq:59}
  \bar{g}_{ij,k} + \bar{g}_{ki,j} + \bar{g}_{jk,i} = 0
\end{equation}
can be obtained from the cyclic Frobenius algebra condition. In 2-component
case there are no additional conditions and the general solution of
compatibility conditions can be obtained starting from arbitrary Monge metric
\begin{align*}
  &\bar{g}_{11} = c_0(u^2)^2 + c_3u^2 + c_4,
  \\
  &\bar{g}_{12}= - c_0u^1u^2 - \frac{1}{2}c_3u^1 -\frac{1}{2}c_1u^2 + c_5,
  \\
  &\bar{g}_{22}= c_0(u^1)^2 + c_1u^1 + c_2.
\end{align*} 
The above metric yields a flat contravariant metric $g^{ij}$, by means
of~\eqref{eq:62}, if and only if the coefficient $c_0$ vanishes. Notice that in
the flat case we recover the metric of the above affine classification Theorem.

It is known \cite{vergallo22:_projec_hamil} that Pl\"ucker embedding provides
an identification of the leading coefficient matrix of a second-order
homogeneous Hamiltonian operator with an algebraic variety, more precisely, a
\emph{linear line congruence}. Such a variety is defined by a system of $n+1$
linear equations in $\mathbb{P}(\wedge^2V)$ and its intersection with Pl\"ucker
variety.

It is then clear that there is a correspondence between trios of Hamiltonian
operators $P_1$, $Q_1$ of the form~\eqref{eq:6} and $R_2$ of the
form~\eqref{eq:36} and trios of algebraic varieties. In the case $n=2$ that is
summarized by the following theorem.
\begin{CorrTh}
  If $n=2$, then there is a bijective correspondence between trios of mutually
  compatible Hamiltonian operators $P_1$, $Q_1$ of the form~\eqref{eq:6} and
  $R_2=\big(\begin{smallmatrix}0 & 1\\-1 & 0\end{smallmatrix}\big)\partial_x^2$
  and pairs of conics $\mathcal{C}_1$, $\mathcal{C}_2$ of rank at least $2$.
\end{CorrTh}
Note that the linear line congruence corresponding to $R_2$ degenerates to $0$
in this case. The theorem is stated and proved as Theorem~\ref{th:classn2}.

In higher dimension the compatibility conditions include Monge's condition on
the metric, but are not reduced to that condition only. However, it easy to
realize that there are plenty of Hamiltonian trios in any dimension.

In particular, using the solver CRACK \cite{WB,WB95}, a package working within
the computer algebra system Reduce \cite{reduce}, we obtain the general
solution $P_1$ of the compatibility conditions $[P_1,R_2]=0$ (with
$R_2=\eta^{ij}\partial_x^2$) for $n=4$. It turns out that there are 288
subcases, each depending on several parameters. Then, for each subcase one
should find all compatible $Q_1$ in the same list. We did this computation for
one selected $P_1$, generalizing the Kaup--Broer and AKNS bi-Hamiltonian pairs,
and obtained a list of $64$ subcases, again each of them depending on several
parameters. Such lists are available upon requests; we just wrote here two
examples of bi-Hamiltonian trios, a local one and a non-local one.

The paper is organized as follows: in Section 2 we briefly recall the canonical
form of second and third order homogeneous operators under projective
reciprocal transformations; in Section 3 we compute the compatibility
conditions between the simplest canonical form of a second order homogeneous
operator and a first order operator of localizable shape; using these results
in Section 4 we study trios of such operators and we provide the classification
in the case $n=2$ and the computational scheme by which we computed the
classification in the case $n=4$, as well as some examples.

\section{Hamiltonian operators and projective reciprocal
  transformations}
\label{sec:preliminaries}

First-order Hamiltonian operators are operators of the form
\begin{equation}
  \label{eq:3}
  P_1 =  g^{ij}(u)\partial_x +\Gamma^{ij}_{k }(u)u^k_x,
\end{equation}
formally skew-adjoint and satisfying the Schouten bracket condition
$[P_1,P_1]=0$.  In the non-degenerate case ($\det(g^{ij})\neq 0$) Dubrovin and
Novikov  proved that $P_1$ is Hamiltonian if and only if $g^{ij}$ is a flat
contravariant pseudo-Riemannian metric and
$\Gamma^j_{hk} = -g_{hi}\Gamma^{ij}_k$ are the Christoffel symbols of the
associated Levi-Civita connection.

Higher-order Dubrovin--Novikov operators have a much more complicated form, see
\cite{DubrovinNovikov:PBHT} for details. However, it was proved
\cite{doyle93:_differ_poiss,potemin97:_poiss,potemin86:_poiss,
  balandin01:_poiss} that if the order is $2$ or $3$, they admit, respectively,
the canonical forms
\begin{equation}
  \label{eq:9}
  R_2 = \partial_x\circ f^{ij}\circ\partial_x, \qquad
  R_3=\partial_x\circ(\ell^{ij}\partial_x + c^{ij}_k u^k_x)\circ\partial_x.
\end{equation}
The above canonical forms are invariant with respect to \emph{projective
  reciprocal transformations} (see \cite{FPV14,vergallo22:_projec_hamil}),
\emph{i.e.} transformations of the type
\begin{equation}
  \label{eq:10}
  d\tilde{x}=\Delta dx,\qquad \tilde{u}^k = \frac{T^k_ju^j + T^k_0}{\Delta},
\end{equation}
where $\Delta = T^0_ju^j + T^0_0$.

The result of transformations \eqref{eq:10} on a Dubrovin--Novikov Hamiltonian
operator is a non-local Hamiltonian operator of localizable shape
\begin{equation}
  \label{eq:12}
  P_1 = g^{ij}(u)\dd_{x} +
  \Gamma^{ij}_{k}(u)u^k_x
  + w^i_k(u)u^k_x\partial_{x}^{-1}u^j_{x}
  + u^i_{x}\partial_{x}^{-1}w^j_k(u)u^k_{x}.
\end{equation}
Operators of this form have been studied in
\cite{F95:_nl_ho,ferapontov03:_recip_hamil} and naturally appear in the study
and classification of integrable systems of PDEs (see for instance
\cite{lorenzoni23:_miura_poiss}). Skew-adjointness and vanishing of the
Schouten bracket in this case lead to the following list of conditions:
\begin{enumerate}
\item $g^{ij}$ is a contravariant pseudo-Riemannian metric and $\Gamma^{ij}_k$
  are the contravariant Christoffel symbols of its Levi-Civita connection;
  equivalently, the following conditions hold:
  \begin{gather}
    \label{eq:46}
    g^{is}\Gamma_s^{jk} = g^{js}\Gamma_{s}^{ik}, \\ \label{eq:47} \partial_k
    g^{ij}=\Gamma^{ij}_k + \Gamma^{ji}_k;
  \end{gather}

\item the following equations hold:
  \begin{gather}\label{eq:4}
    g^{is}w^j_s = g^{js}w^i_s, \\ \label{eq:621} \nabla_i w^j_k = \nabla_kw^j_i,
    \\ \label{eq:13} R^{ij}_{kh} = w^i_k\delta^j_h - w^j_k\delta^i_h -
    w^i_h\delta^j_k + w^j_h\delta^i_k,
  \end{gather}
  where $\nabla$ is the Levi-Civita connection of $g^{ij}$ and
  \begin{equation}\label{eq:43}
    R^{jk}_{sl}=g^{jp}R^k_{psl} = 
    \f{\d\Gamma^{jk}_s}{\d  u^l}-\f{\d\Gamma^{jk}_l}{\d
      u^s}+g_{st}(\Gamma^{tj}_m\Gamma^{mk}_l-\Gamma^{tk}_m\Gamma^{mj}_l)
  \end{equation}
  is the Riemannian curvature tensor of $g_{ij}$.
\end{enumerate}

Canonical forms of operators \eqref{eq:9} under the action of projective
reciprocal transformations have been found in \cite{vergallo22:_projec_hamil}
in the case of second order operators and in \cite{FPV14} and \cite{FPV16} for
third order operators.  The simplest canonical form of second order operators
\eqref{eq:9} is $R_2=\eta^{ij}\partial_x^2$ where $\eta^{ij}$ are the entries
of a constant skew-symmetric matrix.

\section{Conditions of compatibility}
\label{sec:comp-betw-first}

In this Section we calculate the conditions that are equivalent to the
compatibility of $P$ and $R$, \emph{i.e.} the vanishing of the Schouten bracket
$[P,R]=0$, for a pair of Hamiltonian operators, where $P=P_1$ is a non-local
localizable first-order homogeneous Hamiltonian operator as in \eqref{eq:12}
and $R=R_2=\eta^{ij}\partial_x^2$, with $(\eta^{ij})$ a constant skew-symmetric
non-degenerate matrix.

\begin{theorem}\label{th:cond-comp}
  The Hamiltonian operators $P$, $R$ are compatible if and only if
  \begin{itemize}
  \item the functions $w^i_j$ are constant and satisfy the condition
    \begin{gather}
      \label{cond00}
      w^i_l\eta^{lk}+w^k_l\eta^{li}=0;
    \end{gather}
  \item the contravariant Christoffel symbols $\Gamma^{ij}_k$ satisfy the
    conditions:
    \begin{gather}
      \label{cond1}
      \Gamma^{ij}_l\eta^{lk}+\Gamma^{kj}_l\eta^{li}=0,\\
      \label{cond2}
      \Gamma^{ki}_l\eta^{lj}+\Gamma^{ij}_l\eta^{lk}+\Gamma^{jk}_l\eta^{li}=0,\\
      \label{cond3}
      \Gamma^{sj}_{p}\Gamma^{ir}_{s}-\Gamma^{sr}_{p}\Gamma^{ij}_{s}=0,\\
      \label{cond4}
      \f{\d\Gamma^{kj}_l}{\d u^s}=-\delta^j_sw^k_l-w^j_s\delta^k_l.
    \end{gather}
  \end{itemize}
\end{theorem}
\begin{proof}
  We will write differential operators by means of distributions as
  \begin{equation}
    P^{ij}_{xy}=g^{ij}\delta'(x-y)+\Gamma^{ij}_su^s_x\delta(x-y)
    +u^i_x\nu(x-y)w^j_su^s_y+ w^i_su^s_x\nu(x-y)u^j_y\label{eq:38}
  \end{equation}
  and
  \begin{equation}
    R^{ij}_{xy}=\eta^{ij}\delta''(x-y).\label{eq:39}
  \end{equation}
  We use Dubrovin--Zhang formula for the Schouten bracket:
  \begin{align*}
    [P,R]^{ijk}_{x,y,z}=
    &\frac{\partial P^{ij}_{x,y}}{\partial u^{l}(x)}R^{lk}_{x,z}
      + \frac{\partial P^{ij}_{x,y}}{\partial u^{l}(y)} R^{lk}_{y,z}
      +\frac{\partial P^{ki}_{z,x}}{\partial u^{l}(z)} R^{lj}_{z,y}
      + \frac{\partial P^{ki}_{z,x}}{\partial u^{l}(x)} R^{lj}_{x,y}
    \\
    &+\frac{\partial P^{jk}_{y,z}}{\partial u^{l}(y)} R^{li}_{y,x}
      + \frac{\partial P^{jk}_{y,z}}{\partial u^{l}(z)} R^{li}_{z,x}
      + \frac{\partial P^{ij}_{x,y}}{\partial u^{l}_x} \partial_x R^{lk}_{x,z}
      + \frac{\partial P^{ij}_{x,y}}{\partial u^{l}_y} \partial_y R^{lk}_{y,z}
    \\
    &+\frac{\partial P^{ki}_{z,x}}{\partial u^{l}_z} \partial_z R^{lj}_{z,y}
      + \frac{\partial P^{ki}_{z,x}}{\partial u^{l}_x} \partial_x R^{lj}_{x,y}
      +\frac{\partial P^{jk}_{y,z}}{\partial u^{l}_y} \partial_y R^{li}_{y,x}
      + \frac{\partial P^{jk}_{y,z}}{\partial u^{l}_z} \partial_z R^{li}_{z,x}.
  \end{align*}
  The vanishing of the distribution $[P,R]^{ijk}_{x,y,z}$ means that for any
  choice of the test functions $p_i(x),q_j(y),r_k(z)$ the triple integral
  \begin{equation}
    \iiint  [P,R]^{ijk}_{x,y,z} p_i(x)q_j(y)r_k(z)\,dxdydz
  \end{equation}
  should vanish.

  Following \cite{CLV19,lorenzoni04:_biH}, we apply a procedure to collect
  together all terms which are related by a distributional identity. The
  procedure is the following
  \begin{enumerate}
  \item Using identities like
    \begin{equation}
      \nu(z-y)\delta(z-x)=\nu(x-y)\delta(x-z)
    \end{equation}
    together with their differential consequences, we can eliminate all terms
    containing $\nu(z-y)\delta^{(n)}(z-x)$, $\nu(y-x)\delta^{(n)}(y-z)$,
    $\nu(x-z)\delta^{(n)}(x-y)$ producing non-local terms containing
    $\nu(x-y)\delta^{(n)}(x-z)$, $\nu(z-x)\delta^{(n)}(z-y)$,
    $\nu(y-z)\delta^{(n)}(y-x)$ and additional local terms.
  \item Using the identity
    \begin{equation}
      \label{id1}
      f(z)\delta^{(n)}(x-z)=\sum_{k=0}^n\binom{n}{k}f^{(n-k)}(x)\delta^{(n-k)}(x-z),
    \end{equation}
    we can eliminate the dependence on $z$ in the coefficients of
    $\nu(x-y)\delta^{(n)}(x-z)$, the dependence on $y$ in the coefficients of
    $\nu(z-x)\delta^{(n)}(z-y)$ and the dependence on $x$ in the coefficients
    of $\nu(y-z)\delta^{(n)}(y-x)$.  After the first two steps the non-local
    part of $[P,R]^{ijk}_{x,y,z}$ has the form
    \begin{multline}
      a_1(x,y,z)\nu(x-y)\nu(x-z) + \text{cyclic}(x,y,z)\\
      + \sum_{n\ge 0}b_n(x,y)\nu(x-y)\delta^{(n)}(x-z)+\text{cyclic}(x,y,z).
    \end{multline}
  \item The local part of $[P,R]^{ijk}_{x,y,z}$ can be reduced to the form
    \begin{equation}\label{eq:15}
      \sum_{m,n}e_{mn}(x)\delta^{(m)}(x-y)\delta^{(n)}(x-z)
    \end{equation} 
    using the identities (and their differential consequences)
    \begin{equation}
      \delta(z-x)\delta(z-y)=\delta(y-x)\delta(y-z)=\delta(x-y)\delta(x-z)
    \end{equation}
    and the identities~\eqref{id1}.
  \end{enumerate}
  The fulfillment of the Jacobi identity turns out to be equivalent to the
  vanishing of each coefficient in the reduced form. Below a list of relevant
  coefficients in our case.

  The vanishing of the coefficient of $\delta(x-y)\delta'''(x-z)$ provides the
  condition \eqref{cond1}:
  \begin{equation}
    \f{\d g^{jk}}{\d
      u^l}\eta^{li}+\Gamma^{ij}_l\eta^{lk}-\Gamma^{jk}_l\eta^{li}=\Gamma^{ij}_l\eta^{lk}+\Gamma^{kj}_l\eta^{li}=0.\label{eq:383}
  \end{equation}
  The same condition is provided by the vanishing of the coefficient
  $\delta'''(x-y)\delta(x-z)$.

  The vanishing of coefficient of $\delta'(x-y)\delta''(x-z)$ provides the
  condition
  \begin{equation}
    \f{\d g^{ij}}{\d u^l}\eta^{lk}+2\f{\d g^{jk}}{\d
      u^l}\eta^{li}-3\Gamma^{jk}_l\eta^{li}=\Gamma^{ji}_l\eta^{lk}
    +\Gamma^{kj}_l\eta^{li}-\Gamma^{jk}_l\eta^{li}=0,\label{eq:7}
  \end{equation}
  and the vanishing of coefficient of $\delta''(x-y)\delta'(x-z)$ provides the
  condition
  \begin{equation}
    -\f{\d g^{ki}}{\d u^l}\eta^{lj}+\f{\d g^{jk}}{\d
      u^l}\eta^{li}-3\Gamma^{jk}_l\eta^{li}=
    -\Gamma^{ki}_l\eta^{lj}+\Gamma^{kj}_l\eta^{li}-\Gamma^{jk}_l\eta^{li}=0.
    \label{eq:8}
  \end{equation}
  The difference between~\eqref{eq:7} and~\eqref{eq:8} is equivalent to
  condition \eqref{cond1}, while their sum provides \eqref{cond2}:
  \begin{equation}
    \Gamma^{ki}_l\eta^{lj}+\Gamma^{ij}_l\eta^{lk}+\Gamma^{jk}_l\eta^{li}=0.
    \label{eq:29}
  \end{equation}

  The coefficient of $\nu(x-y)\delta'''(x-z)$ is
  \begin{equation}
    w^j_su^s_y\eta^{ik}+w^i_l(x)u^j_y\eta^{lk}+u^j_yw^k_l(x)\eta^{li}+w^j_su^s_y\eta^{ki}=(w^i_l(x)\eta^{lk}+w^k_l(x)\eta^{li})u^j_y;\label{eq:40}
  \end{equation}
  its vanishing is \eqref{cond00}. The same condition is obtained by the
  coefficients of $\nu(z-x)\delta'''(z-y)$ and $\nu(y-z)\delta'''(y-x).$

  The coefficient of $u^s_{xxx}\delta(x-y)\delta(x-z)$ is
  \begin{equation}
    \left(\f{\d\Gamma^{jk}_s}{\d u^l}-\f{\d\Gamma^{jk}_l}{\d
        u^s}\right)\eta^{li} + w^i_l\eta^{lj}\delta^k_s
    +w^k_s\eta^{ij} + w^j_s\eta^{ki} + \delta^j_sw^k_l\eta^{li}.\label{eq:41}
  \end{equation}
  Replacing the condition~\eqref{eq:13} in the previous expression and
  requiring its vanishing we get condition \eqref{cond3}:
  \begin{equation}
    \Gamma^{tj}_m\Gamma^{mk}_l-\Gamma^{tk}_m\Gamma^{mj}_l=0,\label{eq:42}
  \end{equation}
  and the equivalent condition
  \begin{equation}
    \f{\d\Gamma^{jk}_s}{\d  u^l}-\f{\d\Gamma^{jk}_l}{\d
      u^s}=w^j_{l}\delta^k_s+\delta^j_{l}w^k_s-w^j_s\delta^k_l-\delta^j_sw^k_l.
    \label{eq:44}
  \end{equation}

  The coefficient of $\delta(x-y)\delta''(x-z)$ is
  \begin{multline}
    \f{\d \Gamma^{ij}_s}{\d u^l}u^s_x\eta^{lk}+2\d_x\left(\f{\d g^{jk}}{\d
        u^l}\right)\eta^{li}+\f{\d \Gamma^{jk}_s}{\d
      u^l}u^s_x\eta^{li}+u^i_xw^j_l\eta^{lk}
    +w^i_su^s_x\eta^{jk}+u^k_xw^i_l\eta^{lj}
    \\
    +w^k_su^s_x\eta^{ij}
    -3\d_x(\Gamma^{jk}_l)\eta^{li}+3u^j_xw^k_l\eta^{li}+3w^j_su^s_x\eta^{ki}.
  \end{multline}
  Using~\eqref{eq:47} in order to eliminate the derivative of $g^{jk}$, the
  above coefficient can be rewritten as
  \begin{equation}\label{eq:45}
    \left(\f{\d\Gamma^{kj}_l}{\d u^s}
      +\delta^j_sw^k_l+w^j_s\delta^k_l\right)u^s_x\eta^{li}.
  \end{equation}
  Thus the vanishing of this coefficient provides condition \eqref{cond4}.  The
  same condition is provided by the vanishing of the coefficient of
  $\delta''_{xy}\delta_{xz}$.

  The coefficient of $\delta'(x-y)\delta'(x-z)$ is
  \begin{multline}\label{eq:50}
    2\d_x\left(\f{\d g^{jk}}{\d u^l}\right)\eta^{li}+2\f{\d \Gamma^{jk}_s}{\d
      u^l}u^s_x\eta^{li}-u^i_xw^j_l\eta^{lk}-w^i_su^s_x\eta^{jk}+
    \\
    3u^k_xw^i_l\eta^{lj} +3w^k_su^s_x\eta^{ij}-6\d_x(\Gamma^{jk}_l)\eta^{li}
    +3u^j_xw^k_l\eta^{li}+3w^j_su^s_x\eta^{ki}.
  \end{multline}
  It can be proved that the above expression is equal to
  \begin{equation}
    \f{\d}{\d u^s}\left(\Gamma^{kj}_l\eta^{li}+\Gamma^{ik}_l\eta^{lj}
      +\Gamma^{ji}_l\eta^{lk}\right),
  \end{equation}
  and thus vanishes due to condition \eqref{cond2}.

  The coefficient of $\nu(x-y)\delta''(x-z)$ is
  \begin{equation}\label{eq:48}
    u^j_yu^s_x\left(\f{\d w^i_s}{\d u^l}\eta^{lk}-\f{\d w^k_s}{\d u^l}\eta^{li}
      +3\f{\d w^k_l}{\d u^s}\eta^{li}\right),
  \end{equation}
  which is the same as the coefficients of $\nu_{zx}\delta''_{zy}$ and
  $\nu_{yz}\delta''_{yx}$ up to renaming indices and variables.

  The coefficient of $\nu(x-y)\delta'(x-z)$ is
  \begin{equation}\label{eq:49}
    -2u^j_y\d_x\left(\f{\d w^k_s}{\d u^l}u^s_x\right)\eta^{li}
    + 3u^j_y\d_x^2(w^k_l)\eta^{li},
  \end{equation}
  and the same expression, up to renaming indices and variables, holds for the
  coefficients of $\nu_{zx}\delta'_{zy}$ and $\nu_{yz}\delta'_{yx}$.  In the
  expression~\eqref{eq:49}, the coefficient of $u^s_{xx}u^i_y$ is
  \begin{equation}\label{fc}
    -2\f{\d w^k_s}{\d u^l}\eta^{li}+3\f{\d w^k_l}{\d u^s}\eta^{li}.
  \end{equation}

  The coefficient of $\nu(x-y)\delta(x-z)$ is
  \begin{equation}\label{eq:51}
    -u^j_y\d_x^2\left(\f{\d w^k_s}{\d
        u^l}u^s_x\right)\eta^{li}+u^j_y\d_x^3(w^k_l)\eta^{li},
  \end{equation}
  and the same expression holds for the coefficients of $\nu_{zx}\delta_{zy}$
  and $\nu_{yz}\delta_{yx}$ up to renaming indices and variables.  In the
  expression~\eqref{eq:51} the vanishing of the coefficient of $u^j_ju^s_{xxx}$
  provides the closure condition
  \begin{eqnarray*}
    \left(-\f{\d w^k_s}{\d u^l}+\f{\d w^k_l}{\d u^s}\right)\eta^{li}=0.
  \end{eqnarray*}
  Replacing this condition in \eqref{fc} we obtain that the functions $w^i_j$
  are constant.  In particular this tells us that
  \begin{equation}
    \Gamma^{kj}_l=-w^k_lu^j-w^j_su^s\delta^k_l+b^{kj}_l\label{eq:52}
  \end{equation}
  where $b^{kj}_l$ are constant.

  Taking into account this fact the coefficient of $\delta(x-y)\delta'(x-z)$ is
  \begin{multline}\label{eq:53}
    \d_x^2\left(\f{\d g^{jk}}{\d u^l}\right)\eta^{li}+2\d_x\left(\f{\d
        \Gamma^{jk}_s}{\d u^l}u^s_x\right)\eta^{li}+2u^k_{xx}w^i_l\eta^{lj}+
    \\
    +2\d_x(w^k_su^s_x)\eta^{ij}-3\d_x^2(\Gamma^{jk}_l)\eta^{li}+3u^j_{xx}w^k_l\eta^{li}+3\d_x(w^j_su^s_x)\eta^{ki}.
  \end{multline}
  This coefficient vanishes due to previous conditions. Indeed:
  \begin{multline}\label{eq:54}
    \f{\d \Gamma^{kj}_l}{\d u^s}\eta^{li}+2\left(\f{\d \Gamma^{jk}_s}{\d
        u^l}-\f{\d \Gamma^{jk}_l}{\d
        u^s}\right)\eta^{li}+2\delta^k_sw^i_l\eta^{lj}+2w^k_s\eta^{ij}
    +3\delta^j_sw^k_l\eta^{li}+3w^j_s\eta^{ki}=
    \\
    (-\delta^j_sw^k_l-w^j_s\delta^k_l)\eta^{li}+2\left(w^j_{l}\delta^k_s
      +\delta^j_{l}w^k_s-w^j_s\delta^k_l-\delta^j_sw^k_l\right)\eta^{li}+
    \\
    2\delta^k_sw^i_l\eta^{lj}+2w^k_s\eta^{ij}+3\delta^j_sw^k_l\eta^{li}
    +3w^j_s\eta^{ki}=0.
  \end{multline}
  The coefficient of $\delta'(x-y)\delta(x-z)$ is
  \begin{multline}\label{eq:55}
    2\d_x\left(\f{\d \Gamma^{jk}_s}{\d
        u^l}u^s_x\right)\eta^{li}+2u^i_x\d_x(w^j_l)\eta^{lk}+3u^k_{xx}w^i_l\eta^{lj}+
    \\
    3\d_x(w^k_su^s_x)\eta^{ij}-3\d_x^2(\Gamma^{jk}_l)\eta^{li}
    +2u^j_{xx}w^k_l\eta^{li}+2\d_x(w^j_su^s_x)\eta^{ki}.
  \end{multline}
  It vanishes due to previous conditions; the calculation is similar to that
  of~\eqref{eq:54}.

  Finally, the coefficient of $\delta(x-y)\delta(x-z)$ is
  \begin{multline}\label{eq:56}
    \d_x^2\left(\f{\d \Gamma^{jk}_s}{\d
        u^l}u^s_x\right)\eta^{li}+u^k_{xxx}w^i_l\eta^{lj}
    +\d_x^2(w^k_su^s_x)\eta^{ij}-\d_x^3(\Gamma^{jk}_l)\eta^{li}
    \\
    +u^j_{xxx}w^k_l\eta^{li}+\d_x^2(w^j_su^s_x)\eta^{ki}.
  \end{multline}
  Again, this coefficient vanishes due to previous conditions.
\end{proof}

There are very interesting geometric and algebraic consequences of
Theorem~\ref{sec:cond-comp}. First of all, very recently a new algebraic
structure has been introduced in the theory of Integrable Systems, namely
\emph{cyclic Frobenius algebra} \cite{buchstaber23:_cyclic_froben}, in a
framework which is different from ours. It turns out that it also arises in
our context.
\begin{corollary}
  The Christoffel symbols $\Gamma^{ij}_k$ endow the cotangent space $T^*M$ of
  the manifold $M$ of dependent variables $(u^i)$ with a structure of
  \emph{cyclic Frobenius algebra}.
\end{corollary}
\begin{proof}
  The conditions that should be satisfied are exactly \eqref{cond1},
  \eqref{cond2} and \eqref{cond3}.
\end{proof}
An even more surprising fact is the interpretation as an algebraic variety of
the leading coefficient of any first-order non-local homogeneous Hamiltonian
operator $P$ that is compatible with our constant-coefficient second order
Hamiltonian operator $R$~\eqref{eq:39}.
\begin{corollary}
  Let us introduce the nondegenerate symmetric bilinear form
  \begin{equation}
    \label{eq:1}
    \bar{g}_{ab}=\eta_{jb}\eta_{ia}g^{ij}.
  \end{equation}
  Then, $\bar{g}_{ab}$ is the Monge metric of a quadratic line complex.
\end{corollary}
\begin{proof}
  Summing the condition \eqref{cond2} with the same condition with the indices
  $i$, $k$ swapped we obtain the condition
  \begin{equation}
    \label{eq:2}
    g^{ki}_{,l}\eta^{lj} + g^{ij}_{,l}\eta^{lk} + g^{jk}_{,l}\eta^{li} = 0,
  \end{equation}
  where $g^{ki}_{,l}=\pd{g^{ki}}{u^l}$.  The above condition can be rewritten
  in lower indices by multiplication by $\eta_{kb}\eta_{ic}\eta_{ja}$, yielding
  \begin{equation}
    \label{eq:37}
    \bar{g}_{bc,a} + \bar{g}_{ca,b} + \bar{g}_{ab,c} = 0.
  \end{equation}
  The above condition is equivalent to the fact that $\bar{g}_{ab}$ is a Monge
  metric, which is S. Lie's representation of quadratic line complexes (see
  \cite{FPV14,FPV16}). This proves the Corollary.
\end{proof}

\section{Classification of bi-Hamiltonian trios}
\label{sec:class-bi-hamilt-1}

The general problem of the classification of local bi-Hamiltonian trios can be
formulated as follows: classify the bi-Hamiltonian trios of operators of the
form
\begin{equation}
  \label{eq:20}
  A_1 = P_1 + R_2,\qquad A_2 = Q_1,
\end{equation}
where
\begin{itemize}
\item $P_1$, $Q_1$ are local homogeneous first-order Hamiltonian operators;
\item $R_2$ is a local homogeneous second-order Hamiltonian operator;
\item the three operators are mutually compatible:
  \begin{equation}
    [P_1,Q_1]=[R_2,P_1]=[R_2,Q_2]=0.
  \end{equation}
\end{itemize}
Of course, in view of the complexity of the general form of $R_2$, the problem
can be reformulated when $R_2$ is written in the canonical form \eqref{eq:9}.
This can always be done by means of a point transformation of the dependent
variables, without changing the shape of the three operators.

Then, we can use the projective classification of (non-degenerate) second-order
homogeneous operators \cite{vergallo22:_projec_hamil} at the price of allowing
$P_1$ and $Q_1$ to have localizable shape (see
\cite{lorenzoni23:_miura_poiss}). Indeed, the projective classification makes
use of projective reciprocal transformations which transform local operators
into non-local ones.

In this paper, we will just consider the orbit of $R_2$ under the action of
projective reciprocal transformations that contains the constant operator
$R_2^{ij}=\eta^{ij}\partial_x^2$, so to apply the results from the previous
section.

For this reason, we reformulate and restrict the above problem to: classify the
bi-Hamiltonian trios of operators of the form
\begin{equation}
  \label{eq:202}
  A_1 = P_1 + R_2,\qquad A_2 = Q_1,
\end{equation}
where
\begin{itemize}
\item $P_1$, $Q_1$ are non-local homogeneous first-order Hamiltonian operators
  that are localizable (by means of the same projective reciprocal
  transformation);
\item $R_2=\eta^{ij}\partial_x^2$ is a constant-coefficient local homogeneous
  second-order Hamiltonian operator;
\item the three operators are mutually compatible:
  \begin{equation}
    [P_1,Q_1]=[R_2,P_1]=[R_2,Q_2]=0.
  \end{equation}
\end{itemize}

We will be able to give a complete answer in the case $n=2$ and a partial
answer in the case $n=4$, due to the complex structure of the space of
solutions.

We observe that solutions the above version of the problem contain trios of
local operators as a particular case, but they also contain trios where the two
first-order operators cannot be \emph{simultaneously} localized; hence, we
obtain solutions with non-removable non-local terms.

The Hamiltonian operators of our trios are uniquely identified by algebraic
varieties. We now give a brief description of the procedure that allows us to
make the above identification, which, in essence, boils down to Pl\"ucker
embedding.

We assume that $(u^i)$ are affine coordinates of an $n$-dimensional projective
space $\mathbb{P}(V)$, where $V$ is a vector space with $\dim V=n+1$ and
coordinates $(v^i)$.  Homogeneous coordinates on $\mathbb{P}(V)$ are denoted by
$[v^1,$ $\ldots,v^{n+1}]$, in such a way that $u^i=v^i/v^{n+1}$.  We recall
that Pl\"ucker embedding (of lines) is the natural injective map
$\operatorname{Gr}(2,V)\hookrightarrow \mathbb{P}(\wedge^2V)$, where
$\operatorname{Gr}(2,V)$ is the Grassmanniann of planes in $V$, which can be
identified as the space of projective lines in $\mathbb{P}(V)$.

Elements of $\mathbb{P}(\wedge^2V)$ can be represented as $[p^{ij}]$, where
$p^{ij}$ are coordinates with respect to the basis $e_i\wedge e_j$, $i<j$, of
$\wedge^2 V$, $(e_i)$ being a basis of $V$. The coordinates $p^{ij}$ are
Pl\"ucker coordinates.

The image of Pl\"ucker embedding can be characterized as the space of of
decomposable forms in $\wedge^2V$; it is an algebraic variety described by
a system of homogeneous quadratic relations between Pl\"ucker coordinates:
$p^{ij}p^{kh} - p^{ik}p^{jh} + p^{ih}p^{jk} = 0$, where $i<j<k<h$. The system
is empty if $n=2$, consists of one quadric only if $n=3$, $5$ quadrics if
$n=4$, etc..

A single, additional quadratic equation $X^T\mathcal{Q}X=0$, where $X=(p^{ij})$
and $\mathcal{Q}$ is a symmetric matrix of order
$\dim\wedge^2V=\binom{n+1}{2}$, together with the equations that define
Pl\"ucker variety is a quadratic line complex.

The lines of the quadratic line complex passing through a single point $x$ in
the projective space form a quadratic cone. This $x$-dependent family of cones
endows the projective space with a conformal structure, the Monge metric.  The
Monge metric is obtained by considering lines through two infinitesimally close
points $P$, with coordinates $[v^1,\ldots,v^{n+1}]$, and $P+dP$, with
coordinates $[v^1+dv^1,\ldots,v^{n+1}+dv^{n+1}]$. Then, the Pl\"ucker
coordinates are the minors $p^{ij}=v^idv^j - v^jdv^i$, with $i,j=1,\ldots,n+1$,
$i<j$, of the matrix
\begin{equation}
  \label{eq:60}
  \begin{pmatrix}
    v^1 & \cdots & v^{n+1}\\
    v^1 +dv^1& \cdots & v^{n+1}+dv^{n+1}
  \end{pmatrix}.
\end{equation}
In affine coordinates, upon substituting $v^{n+1}=1$, $dv^{n+1}=0$, the Monge
metric is a quadratic form with respect to the one-forms
\begin{equation}
  \label{eq:22}
  u^idu^j - u^jdu^i,\quad i<j,\qquad du^i
\end{equation}
(modulo Pl\"ucker variety); its coefficients are quadratic polynomials (but
such a condition is not enough to characterize Monge metrics). The above
geometric construction has been exploited by many geometers in the past, like
A. Clebsch, S. Lie and C. Segre, but has been forgotten until recently (see
\cite{FPV14,FPV16} and the history paper \cite{rogora}).

From the above discussion, it is easy to generate an ansatz for a first-order
operator $P_1$ that is compatible with a constant-coefficient second-order
operator $R_2$, using the formula \eqref{eq:1} and a generic Monge metric
$\bar{g}_{ij}$.

We remark that also $R_2$ defines a projective variety in the same space as the
above quadratic line complex, according with the identification in
\cite{vergallo22:_projec_hamil}. More precisely, the two-form
$\eta_{ij}du^i\wedge du^j$ can be made into a three-form
$\eta_{ij\,n+1}dv^i\wedge dv^j\wedge dv^{n+1}$, where
$\eta_{ij\,n+1}=\eta_{ij}$, and this yields an algebraic variety in
$\mathbb{P}(\wedge^2V)$ defined by the equations $\eta_{ijk}p^{jk}=0$ and
Pl\"ucker's variety equations (here $\eta_{ijk}$ is obtained from
$\eta_{ij\,n+1}=\eta_{ij}$ by skew-symmetrization). Such a variety is a
\emph{linear line congruence}. We will discuss it in the case $n=4$.

\subsection{Case \texorpdfstring{$n=2$}{n=2}: classification}
\label{sec:case-n=2}
\begin{theorem}\label{th:classn2}
  Let $R_2$ and $P_1$ be Hamiltonian operators of the following shape:
  \begin{equation}
    \label{eq:21}
    R_2=\begin{pmatrix}0 & 1\\ -1 & 0\end{pmatrix},\qquad
    P_1 = g^{ij}\partial_x + \Gamma^{ij}_ku^k_x +
    w^i_ku^k_x\partial_x^{-1}u^j_x +
    u^i_x\partial_x^{-1}w^j_hu^h_x .
  \end{equation}
  Then, the following conditions are equivalent:
  \begin{itemize}
  \item $[R_2,P_1]=0$;
  \item The local part of $P_1$ is determined by an arbitrary non-degenerate
    Monge metric $(\bar{g}_{ab})$ through the formula~\eqref{eq:1}. More
    explicitely, we have
    \begin{equation}
      \begin{split}
        &g^{11}=c_0(u^1)^2 + c_1u^1+c_2,
        \\
        &g^{12}=c_0u^1u^2 + \frac{1}{2}c_3u^1+\frac{1}{2}c_1u^2+c_5,
        \\
        &g^{22}=c_0 (u^2)^2 + c_3u^2+c_4,
      \end{split}\label{eq:23}
    \end{equation}
    where $c_0$, $c_1$, $c_2$, $c_3$, $c_4$, $c_5$ are arbitrary parameters.
    The non-local part of $P_1$ is given by $(w^i_j)=-1/2c_0\operatorname{Id}$,
    hence the operator has the form
    \begin{equation}
      \label{eq:26}
      P_1 = g^{ij}\partial_x + \Gamma^{ij}_k u^k_x - c_0u^i_x\partial_x^{-1}u^j_x.
    \end{equation}
  \end{itemize}
\end{theorem}
\begin{proof}
  The unknown metric $g^{ij}$ can be reconstructed from a Monge metric using
  \eqref{eq:1}. In this way, \eqref{cond2} will be solved by
  construction. Then, a simple calculation proves that the equations
  \eqref{cond1} and \eqref{cond3} are verified.

  From \eqref{cond00} we easily obtain $w^1_2=w^2_1=0$ and $w^1_1=w^2_2$. If we
  use such conditions, all other equations are identically verified, with the
  exception of~\eqref{cond4} that yields the equation
  \begin{equation}\label{eq:27}
    c_0=-2w^1_1=-2w^2_2 .
  \end{equation}
  Concerning the Hamiltonian conditions on $P_1$, we see that~\eqref{eq:621} is
  verified by the contravariant metric~\eqref{eq:23} and
  $w^i_j=-1/2 c_0\delta^i_j$. Moreover, it is easy to calculate that the only
  nonzero component of the curvature $(g_{ij})$ is $R^{12}_{12}=-c_0$; using
  the condition~\eqref{eq:27}, we immediately see that~\eqref{eq:13} is
  verified.

  The Monge metric of the operator $P_1$ is
  \begin{align*}
    &\bar{g}_{11} = c_0(u^2)^2 + c_3u^2 + c_4,
    \\
    &\bar{g}_{12}= - c_0u^1u^2 - \frac{1}{2}c_3u^1 -\frac{1}{2}c_1u^2 + c_5,
    \\
    &\bar{g}_{22}= c_0(u^1)^2 + c_1u^1 + c_2.
  \end{align*}
  It is easy to prove that the (symmetric) matrix $\mathcal{Q}$ of the
  corresponding quadratic line complex is generic (up to the non-degeneracy
  requirement): if we fix Lie's form of Pl\"ucker's coordinates
  \begin{equation}
    \label{eq:25}
    X^T= (u^1du^2 - u^2du^1, du^1, du^2)
  \end{equation}
  the Monge metric $\bar{g}$ is the quadratic expression
  $\bar{g}=X^T\mathcal{Q}X$ where
  \begin{equation}
    \label{eq:24}
    \mathcal{Q}=
    \begin{pmatrix}
      c_0 & -\frac{1}{2}c_3 & \frac{1}{2}c_1\\
      -\frac{1}{2}c_3 & c_4 & c_5 \\
      \frac{1}{2}c_1 & c_5 & c_2\\
    \end{pmatrix}
  \end{equation}
  This is a generic conic in $\mathbb{P}(V)$ (up to the non-degeneracy
  requirement on $(g^{ij})$).
\end{proof}

\begin{corollary}
  The Hamiltonian operator $P_1$ is local if and only if $c_0=0$; in this case,
  the operator coincides with the class that has been found in
  \cite{LSV:bi_hamil_kdv}.
\end{corollary}
Note that locality is not preserved by projective reciprocal transformations.

We are ready to state the Projective Correspondence Theorem.
\begin{theorem}[Projective Correspondence Theorem]
  Let $n=2$. Then, a trio of mutually compatible Hamiltonian operators $P_1$,
  $Q_1$, $R_2$ of the form~\eqref{eq:202} is equivalently given by any two
  conics $\mathbb{C}_1$, $\mathbb{C}_2$ in $\mathbb{P}(V)$, each of rank at
  least $2$.
\end{theorem}
\begin{proof}
  We observe that the action of projective reciprocal transformations on $R_2$
  yields $R_2$ multiplied by the determinant of the projective transformation,
  so $R_2$ is invariant under the action of $SL(V)$.

  Then, the action of $SL(V)$ on $V$ induces an action on $\wedge^2V$ that, in
  the case $n=2$, is bijective on $SL(\wedge^2V)$. This means that conics in
  $\mathbb{P}(\wedge^2V)$ can be classified by their rank (provided we regard
  $V$ as a complex vector space).

  The rank of the quadratic line complex corresponding to a non-degenerate
  Monge metric must be at least $2$; a rank $1$ quadratic line complex yields a
  degenerate Monge metric.

  Finally, we observe that any pencil $P_1+\lambda Q_1$ of operators of the
  type \eqref{eq:23} is of operators of the same type, due to the linearity of
  the coefficients. That implies that any two operators whose metric is defined
  by~\eqref{eq:23} are compatible.
\end{proof}

\begin{remark}
  When $n=2$ the Pl\"ucker variety is empty. Moreover, it is immediate to prove
  that the algebraic variety defined by $R_2$, a linear line complex,
  degenerates to $0$. So, no other algebraic variety else than the two conics
  of the above statements come into play when $n=2$.
\end{remark}

A first projective classification of bi-Hamiltonian trios of the shape of
Theorem~\ref{th:classn2} can be made in the following way.
\begin{proposition}
  With respect to the action of projective reciprocal transformations, there
  are two inequivalent classes of trios $R_2$, $P_1$, $Q_1$ that are mutually
  compatible and of the type \eqref{eq:21}. They are described by
  \begin{enumerate}
  \item $R_2$, $P_{1,2}$, $Q_1$, where the quadratic line complex corresponding
    to $P_{1,2}$ has rank $2$ and $Q_1$ is arbitrary, and
  \item $R_2$, $P_{1,3}$, $Q_1$, where the quadratic line complex corresponding
    to $P_{1,3}$ has rank $3$ and $Q_1$ is arbitrary.
  \end{enumerate}
\end{proposition}
This classification is far from being complete; indeed, finding the invariants
of a pair of quadratic forms is a well-known problem, its first solution is due
to Weierstrass. Then, the normal forms yield a complete classification of the
pairs. We will get back to this topic in a forthcoming extended version of this
paper.

\subsection{Case \texorpdfstring{$n=2$}{n=2}: examples}
\label{sec:example}

We will consider, as the simplest example in $n=2$ components, the Poisson
pencil of the Kaup--Broer system (first obtained in \cite{Kuper_85}). The trio
is defined by
\begin{gather}
  \label{eq:111}
  P_1 =
  \begin{pmatrix}
    0 & \partial_x \\ \partial_x & 0
  \end{pmatrix}, \quad P_2 =
  \begin{pmatrix}
    2\partial_x & \partial_x u^1 - \partial_x^2
    \\
    u^1\partial_x + \partial_x^2 & u^2\partial_x + \partial_x u^2
  \end{pmatrix},
  \\
  R=
  \begin{pmatrix}
    0 & -1\\ 1 & 0
  \end{pmatrix}\partial_x^2.
\end{gather}
The first-order operators have the leading coefficient matrices
\begin{equation}
  \label{eq:11}
  (g^{ij}_1) =
  \begin{pmatrix}
    0 & 1\\ 1 & 0
  \end{pmatrix},\qquad (g^{ij}_2) =
  \begin{pmatrix}
    2 & u^1\\ u^1 & 2u^2
  \end{pmatrix}.
\end{equation}
The corresponding Monge metrics are
\begin{equation}
  \label{eq:14}
  (\bar{g}_{1,ab}) =
  \begin{pmatrix}
    0 & -1\\ 1 & 0
  \end{pmatrix}^T
  \begin{pmatrix}
    0 & 1\\ 1 & 0
  \end{pmatrix}
  \begin{pmatrix}
    0 & -1\\ 1 & 0
  \end{pmatrix}=
  \begin{pmatrix}
    0 & -1\\ -1 & 0
  \end{pmatrix}
\end{equation}
and
\begin{equation}
  \label{eq:153}
  (\bar{g}_{2,ab}) =
  \begin{pmatrix}
    0 & -1\\ 1 & 0
  \end{pmatrix}^T
  \begin{pmatrix}
    2 & u^1\\ u^1 & 2u^2
  \end{pmatrix}
  \begin{pmatrix}
    0 & -1\\ 1 & 0
  \end{pmatrix}=
  \begin{pmatrix}
    2u^2 & -u^1\\ -u^1 & 2
  \end{pmatrix}.
\end{equation}
We recall that Pl\"ucker's coordinates in Monge form are
\begin{equation}
  \label{eq:16}
  u^1du^2 - u^2du^1,\quad du^1, \quad du^2.
\end{equation}
With respect to the above coordinates, the matrices of the quadratic line
complexes take the form
\begin{equation}
  \label{eq:17}
  Q_1 =
  \begin{pmatrix}
    0 & 0 & 0\\ 0 & 0 & 1\\ 0 & 1 & 0
  \end{pmatrix},\quad Q_2 =
  \begin{pmatrix}
    0 & -1 & 0\\ -1 & 0 & 0\\ 0 & 0 & 2
  \end{pmatrix}.
\end{equation}
Indeed, it is easy to realize that
\begin{equation}
  \label{eq:18}
  \bar{g}_{2,ab}du^adu^b = -2(u^1du^2 - u^2du^1)du^1 + 2du^2du^2,
\end{equation}
and similarly for the other Monge metric. We observe that Pl\"ucker variety is
empty for the Pl\"ucker embedding of $\mathbb{P}^2$, so the above quadratic
forms provide the only defining equations for the corresponding quadratic line
complexes.

\begin{remark}
  Note that $\rk(Q_1)=2$ and $\rk(Q_2)=3$. That means that, while $Q_1$ defines
  a third-order homogeneous Hamiltonian operator according with \cite{FPV14},
  $Q_2$ \emph{does not} define a local third-order HHO (but see
  \cite{casati19:_hamil}, as it could be non-local!).
\end{remark}

\begin{remark}
  Another remarkable example is the AKNS Hamiltonian trio (see
  \cite{LSV:bi_hamil_kdv} and references therein); we will not calculate the
  corresponding quadratic line complexes here as they can be found as in the
  above Example; however, both first-order operators are defined by a Monge
  metric whose matrix $Q$ has rank $2$: $Q_1$ is the same as in the previous
  example and the other is
  \begin{equation}
    \label{eq:19}
    Q_2 =
    \begin{pmatrix}
      0 & 1 & 0\\ 1 & 2 & 0\\ 0 & 0 & 0
    \end{pmatrix}.
  \end{equation}
\end{remark}

\subsection{Case \texorpdfstring{$n=4$}{n=4}: classification}
\label{sec:class-bi-hamilt}

When $n=4$, we have been able to find a complete solution of the problem. We
used the following algorithm.

First of all, we fix a second-order operator, for example
\begin{equation}
  \label{eq:28}
  R_2=
  \begin{pmatrix}
    0 & 0 & 0 & -1\\ 0 & 0 & -1 & 0\\ 0 & 1 & 0 & 0\\ 1 & 0 & 0 & 0
  \end{pmatrix}\partial_x^2.
\end{equation}
We observe that the corresponding $3$-form is
\begin{displaymath}
  \eta = -2dv^1\wedge dv^4\wedge dv^5 -2dv^2\wedge dv^3\wedge dv^5.
\end{displaymath}
The equations of the corresponding linear line congruence are
$\eta_{ijk}p^{jk}=0$, which translate into the system
\begin{equation}
  \label{eq:61}
  p^{14}+p^{23} = 0,\ p^{i5}=0,\, i<5,
\end{equation}
which yield a linear line congruence upon intersection with
$\operatorname{Gr}(2,V)$.

Reciprocal projective transformations act non-trivially on $R_2$, but we will
know all canonical forms of the trios if we compute $P_1$ as the non-local and
localizable first-order homogeneous operators that are compatible with $R_2$:
$[R_2,P_2]=0$. In this subspace of operators we must then compute all pairs
$P_1$, $Q_1$ of operators that are compatible: $[P_1,Q_1]=0$ to form trios (see
\eqref{eq:202} and thereafter).

We brought to an end the first part of the above programme: we computed all
$P_1$ that are compatible with $R_2$ and have the above form. The results are
available upon request by email. The calculation was nontrivial
and was performed on a compute server of the Istituto Nazionale di Fisica
Nucleare (INFN -- Italian National Institute of Nuclear Physics), using Reduce
\cite{reduce,reduceproject} and about 64GB of RAM for 1h.

It is worth to describe the algorithm that we used.
\begin{enumerate}
\item First of all, since we know that the metric of the first-order operator
  is a Monge metric, we calculate the most general Monge metric in the case
  $n=4$. It is parametrized by a finite number of constants.
\item We also know that $w^i_j$ are constants, and we use this information in
  the setup of the computation.
\item Christoffel symbols $\Gamma^{ij}_k$ are determined by the formula
  \eqref{eq:52} in terms of $w^i_j$ and of new unknown constants $b^{ij}_k$.
  Summarizing, the unknowns are constants, and are: the coefficients in the
  Monge metric, the coefficients in the `tail' $w^i_j$ and the coefficients
  $b^{ij}_k$ that make up $\Gamma^{ij}_k$.
\item Then, compatibility equations are solved. There are $2$ groups of linear
  equations in the above unknowns: \eqref{cond00} and \eqref{cond1}. The
  conditions \eqref{cond2} and \eqref{cond4} are automatically satisfied.  The
  nonlinear condition is the associativity condition \eqref{cond3}.

  The Hamiltonian operator conditions on $P_1$ are \eqref{eq:47} (which is
  linear with respect to the unknowns), \eqref{eq:46}, \eqref{eq:4},
  \eqref{eq:621} (which are nonlinear). Note that the equations \eqref{eq:13}
  are automatically fulfilled.
\item The overdetermined system solver CRACK \cite{WB,WB95}, a package working
  in Reduce, was used to solve the above nonlinear algebraic equations. The
  solution is way too involved to be printed out here, consisting of 288
  subcases. It is available on demand by email.
\item It is excessively complicated to write down all solutions of the
  compatibility conditions from $[P_1,Q_1]=0$, where $P_1$ and $Q_1$ are two
  solutions of the above equations (think of each of the 288 subcases to be
  used in a compatibility computation with another operator from each of the
  $288$ subcases). However, the solutions are computable in reasonable time
  with modern computers, see below.
\end{enumerate}

We observe that the results obtained are not exactly a classification of the
trios with the given $R_2$; indeed, the reciprocal transformations act on $R_2$
with a stabilizer, that might be used to reduce the number of constants in
$P_1$. At the moment, we do not consider this problem.

\subsection{Case \texorpdfstring{$n=4$}{n=4}: a subclass}
\label{sec:case-n=4:-subclass}

In view of the complexity of the compatibility calculation of the operators in
the full set of solutions of $[R_2,P_1]=0$, we can present here the results for
a subset of all possible trios: namely, those that are a direct generalization
of the Kaup--Broer and the AKNS trios in Subsection~\ref{sec:example}.

Indeed, we can observe that in those examples $P_1$ has always constant form
(in particular, its matrix is the `antidiagonal identity'). We can therefore
postulate the form of $P_1$ (besides the form of $R_2$) as
\begin{equation}
  \label{eq:30}
  P_1 =
  \begin{pmatrix}
    0 & 1 & 0 & 0
    \\
    1 & 0 & 0 & 0
    \\
    0 & 0 & 0 & 1
    \\
    0 & 0 & 1 & 0
  \end{pmatrix}\partial_x
\end{equation}
and then find, in the set of solutions of $[R_2,Q_1]=0$ with the given ansatz
of $Q_1$, those that are compatible with $P_1$: $[P_1,Q_1]=0$.

We obtain 64 cases of first-order operators $Q_1$ as above. The computation is
shorter than that of $[R_2,Q_2]=0$ only, and can be done on a modern laptop. It
makes use of the packages developed in \cite{m.20:_weakl_poiss} in order to
calculate the conditions $[P_1,Q_1]=0$. Here, we will just show two cases, one
is local and the other is non-local.

\subsubsection*{Local case}

The metric of the first-order operator is
\footnotesize\begin{equation}
  \label{eq:31}
  (g^{ij})= \begin{pmatrix}
    2  b^{11}_2  u^{2} + c_{55} & c_{54} & b^{11}_2  u^{4}
    +b^{13}_1  u^{1} - c_{49} & b^{13}_1  u^{2} - c_{34}
    \\ 
    c_{54} & 0 &b^{13}_1  u^{2} - c_{34} & 0
    \\ 
    b^{11}_2  u^{4} + b^{13}_1  u^{1} - c_{49} &
    b^{13}_1  u^{2} - c_{34} & 2  b^{13}_1  u^{3} + c_{46} &
    2  b^{13}_1  u^{4} + c_{31}
    \\
    b^{13}_1  u^{2} - c_{34} & 0 & 2  b^{13}_1  u^{4} + c_{31} & 0
  \end{pmatrix}
\end{equation}
\normalsize The free parameters are
\begin{equation}
  \label{eq:35}
  b^{11}_2,b^{13}_1,c_{31},c_{34},c_{46},c_{49},c_{54},c_{55}.
\end{equation}
Nonzero coefficients in the Christoffel symbols are determined by the only
nonzero constants $b^{ij}_k$, which are
\begin{gather*}
  b^{14}_2=b^{13}_1,\quad b^{23}_2=b^{13}_1,\quad b^{31}_4=b^{11}_2,
  \\
  b^{33}_3=b^{13}_1,\quad b^{34}_4=b^{13}_1, \quad b^{43}_4=b^{13}_1.
\end{gather*}
It turns out that nonzero Christoffel symbols (in upper indices) are
\begin{gather*}
  \Gamma^{11}_2=b^{11}_2,\quad \Gamma^{13}_1=b^{13}_1,\quad
  \Gamma^{14}_2=b^{13}_1,\quad \Gamma^{23}_2=b^{13}_1,
  \\
  \Gamma^{31}_4=b^{11}_2,\quad \Gamma^{33}_3=b^{13}_1, \quad
  \Gamma^{34}_4=b^{13}_1, \quad \Gamma^{43}_4=b^{13}_1.
\end{gather*}

\subsubsection*{Non-local case}
The metric of the first-order operator is
\begin{multline}
  \label{eq:32}
  (g^{ij})=\left(
  \begin{matrix}
    0&c_{54} -(u^{1})^{2} w^2_1
    \\
    c_{54}-(u^{1})^{2} w^2_1&2 b^{22}_1 u^{1} +c_{53}-2 u^{1} u^{2} w^2_1 
    \\
    0& - (c_{34} + u^{1} u^{3} w^2_1) 
    \\
    -(c_{34} +u^{1} u^{3} w^2_1) &b^{22}_1 u^{3} - c_{33}- u^{1} u^{4} w^2_1
    -u^{2} u^{3} w^2_1
  \end{matrix}
\right.
\\
\left.
  \begin{matrix}
    0 & - (c_{34} +u^{1} u^{3} w^2_1 )
    \\
     - (c_{34} +u^{1} u^{3} w^2_1 ) & b^{22}_1 u^{3} -c_{33}
     -u^{1} u^{4} w^2_1 - u^{2} u^{3} w^2_1
     \\
      0 & c_{31} - (u^{3})^{2} w^2_1
     \\
     c_{31} - (u^{3})^{2} w^2_1& c_{28} - 2 u^{3} u^{4} w^2_1
  \end{matrix}
  \right)
\end{multline}
The non-local part is defined by the free parameter $w^2_1$ (with the
requirement $w^2_1\neq 0$) and the equations
\begin{equation}
  \label{eq:34}
  w^4_3=w^2_1,\qquad w^i_j=0 \quad\text{otherwise.}
\end{equation}

The free parameters are
\begin{equation}
  \label{eq:33}
  b^{22}_1,w_{21},c_{28},c_{31},c_{33},c_{34},c_{53},c_{54}
\end{equation}

The only nonzero constants $b^{ij}_k$ are
\begin{equation}
  \label{eq:58}
  b^{22}_1,\quad b^{42}_3=b^{22}_1.
\end{equation}

The nonzero Christoffel symbols are
\begin{gather*}
  \Gamma^{12}_1= - u^1 w^2_1,\quad \Gamma^{14}_1= - u^3w^2_1,\quad
  \Gamma^{21}_1= - u^1w^2_1,\quad \Gamma^{22}_1=b^{22}_1 - u^2w^2_1,
  \\
  \Gamma^{22}_2= - u^1w^2_1,\quad \Gamma^{23}_1= - u^3w^2_1\quad \Gamma^{24}_1=
  - u^4w^2_1,\quad \Gamma^{24}_2= - u^3w^2_1,
  \\
  \Gamma^{32}_3= - u^1w^2_1,\quad \Gamma^{34}_3= - u^3w^2_1,\quad
  \Gamma^{41}_3= - u^1w^2_1,\quad \Gamma^{42}_3=b^{22}_1 - u^2w^2_1
  \\
  \Gamma^{42}_4= - u^1w^2_1,\quad \Gamma^{43}_3= - u^3w^2_1,\quad
  \Gamma^{44}_3= - u^4w^2_1,\quad \Gamma^{44}_4= - u^3w^2_1.
\end{gather*}

\bigskip

\textbf{Acknowledgements.} We thank P. Antonini, R. Chiriv\`\i, S. Opana\-senko
for useful discussions.

\end{document}